\documentclass{article}

\usepackage{float}
\usepackage{amsthm,,amssymb, color, epsfig, graphics}
\usepackage{xcolor}
\usepackage{graphicx}
\usepackage[intlimits]{amsmath}
\usepackage{latexsym}
\usepackage{accents}

\newtheorem{theorem}{Theorem}

\newtheorem{observation}[theorem]{Observation}
\newtheorem{corollary}[theorem]{Corollary}

\theoremstyle{definition}

\newtheorem{remark}{Remark}

\def\beq{\begin{equation}}
\def\eeq{\end{equation}}
\def\bea{\begin{eqnarray}}
\def\eea{\end{eqnarray}}

\makeatletter
\expandafter\let\expandafter
\reset@font\csname reset@font\endcsname
\def\subeqnarray{\arraycolsep1pt
   \def\@eqnnum\stepcounter##1{\stepcounter{subequation}
       {\reset@font\rm(\theequation\alph{subequation})}}
\jot5mm     \eqnarray}

\makeatother

\newcommand{\bbZ}{{\mathbb Z}}
\newcommand{\bbR}{{\mathbb R}}

\def\ep{\epsilon}

\def\su2{{\mathfrak {su}}(2)}
\def\e3{{\mathfrak {e}}(3)}

\DeclareMathAccent{\wtilde}{\mathord}{largesymbols}{"65}
\newcommand{\untilde}[1]{\underaccent{\wtilde}{#1}}

\newcommand{\wx}{{\widetilde{x}}}
\newcommand{\wm}{{\widetilde{m}}}
\newcommand{\wip}{{\widetilde{p}}}
\newcommand{\wwm}{{\widetilde{\widetilde{m}}}}
\newcommand{\wwp}{{\widetilde{\widetilde{p}}}}
\newcommand{\um}{\untilde{m}}
\newcommand{\up}{\untilde{p}}

\begin{document}

%

%
\thispagestyle{empty}

\title{New results on integrability \\ of the Kahan-Hirota-Kimura  discretizations}

\author{Matteo Petrera~$^a$ and Yuri B. Suris~$^b$}

\maketitle

\footnotetext{Institut f\"ur Mathematik, MA 7-1, 
Technische Universit\"at Berlin, Str. des 17. Juni 136, 
10623 Berlin, Germany. \\
$^a$ E-mail: petrera@math.tu-berlin.de,
$^b$ E-mail: suris@math.tu-berlin.de}

\begin{abstract}
  \noindent
  R.~Hirota and K.~Kimura discovered integrable discretizations of
the Euler and the Lagrange tops, given by birational maps. Their
method is a specialization to the integrable context of a general
discretization scheme introduced by W.~Kahan and applicable to any
vector field with a quadratic dependence on phase variables.

We report several novel observations regarding integrability of the Kahan-Hirota-Kimura discretization. For several of the most complicated cases for which integrability is known  (Clebsch system, Kirchhoff system, and Lagrange top),
\begin{itemize}
\item we give nice compact formulas for some of the more complicated integrals of motion and for the density of the invariant measure, and
\item we establish the existence of higher order Wronskian Hirota-Kimura bases, generating the full set of integrals of motion.
\end{itemize}
While the first set of results admits nice algebraic proofs, the second one relies on computer algebra.
\end{abstract}

\section{Introduction}
\label{Sect: intro}

The Kahan-Hirota-Kimura discretization method was introduced in the geometric integration literature by Kahan in the unpublished notes \cite{K} as a method applicable to any system of ordinary differential equations with a quadratic vector field:
\begin{equation}\label{eq: diff eq gen}
\dot{x}=f(x)=Q(x)+Bx+c,
\end{equation}
where each component of $Q:\bbR^n\to\bbR^n$ is a quadratic form, while $B\in{\rm Mat}_{n\times n}(\bbR)$ and $c\in\bbR^n$. Kahan's discretization (with stepsize $2\epsilon$) reads as
\begin{equation}\label{eq: Kahan gen}
\frac{\widetilde{x}-x}{2\epsilon}=Q(x,\widetilde{x})+\frac{1}{2}B(x+\widetilde{x})+c,
\end{equation}
where
\[
Q(x,\widetilde{x})=\frac{1}{2}\big(Q(x+\widetilde{x})-Q(x)-Q(\widetilde{x})\big)
\]
is the symmetric bilinear form corresponding to the quadratic form $Q$. We say that the expression on the right-hand side of \eqref{eq: Kahan gen} is the {\em polarization} of the expression on the right-hand side of \eqref{eq: diff eq gen}. Equation (\ref{eq: Kahan gen}) is {\em linear} with respect to $\widetilde x$ and therefore defines a {\em rational} map $\widetilde{x}=\Phi_f(x,\epsilon)$. Clearly, this map approximates the time $2\epsilon$ shift along the solutions of the original differential system. Since equation (\ref{eq: Kahan gen}) remains invariant under the interchange $x\leftrightarrow\widetilde{x}$ with the simultaneous sign inversion $\epsilon\mapsto-\epsilon$, one has the {\em reversibility} property
\begin{equation}\label{eq: reversible}
\Phi_f^{-1}(x,\epsilon)=\Phi_f(x,-\epsilon).
\end{equation}
In particular, the map $\Phi_f$ is {\em birational}. 

Kahan applied this discretization scheme to the famous Lotka-Volterra system and showed that in this case it possesses a very remarkable non-spiralling property. This property was explained by Sanz-Serna \cite{SS} by demonstrating that in this case the numerical method preserves an invariant Poisson structure of the original system.

The next intriguing appearance of this discretization was in two papers by Hirota and Kimura who (being apparently unaware of the work by Kahan) applied it to two famous {\em integrable} system of classical mechanics, the Euler top and the Lagrange top \cite{HK, KH}. Surprisingly, the discretization scheme produced in both cases {\em integrable} maps. 

In \cite{PS, PPS1, PPS2} the authors undertook an extensive study of the properties of the Kahan's method when applied to integrable systems (we proposed to use in the integrable context the term ``Hirota-Kimura method''). It was demonstrated that, in an amazing number of cases, the method preserves integrability in the sense that the map $\Phi_f(x,\epsilon)$ possesses as many independent integrals of motion as the original system $\dot x=f(x)$.

Further remarkable geometric properties of the Kahan's method were discovered by Celledoni, McLachlan, Owren and Quispel in \cite{CMOQ1}, see also \cite{CMOQ2, CMOQ3}. These properties are unrelated to integrability. They demonstrated that for an arbitrary Hamiltonian vector field with a constant Poisson tensor and a cubic Hamilton function, the map $\Phi_f(x,\epsilon)$ possesses a rational integral of motion and an invariant measure with a polynomial density.

The goal of the present paper is to communicate several novel observations regarding integrability of the Kahan's method. These observations hold for several of the most complicated cases for which integrability of the Kahan-Hirota-Kimura discretization is established  (Clebsch system, $so(4)$ Euler top, Kirchhoff system, and Lagrange top). However, some of our new findings here are verifiable by hands and do not require heavy computer algebra computations. This refers to nice compact formulas for some of the more complicated integrals of motion and for the density of the invariant measure. See Theorem \ref{th volume} and Observation \ref{conjecture 1} in Section \ref{Sect: novel}. We give these results for all of the above mentioned systems, but provide detailed proofs for the first Clebsch flow only (see Section \ref{sect Clebsch 1}). Another set of results still relies on the computer algebra. This refers to the so called higher order Wronskian Hirota-Kimura bases. See Observation \ref{conjecture 2} in Section \ref{Sect: novel}. We expect that understanding of the latter phenomenon could be crucial for the whole integrability picture of the Kahan-Hirota-Kimura discretizations, but to this moment the origin of this phenomenon remains obscure.
 
\section{General properties of Kahan-Hirota-Kimura discretization}

The explicit form of the map $\Phi_f$ defined by \eqref{eq: Kahan gen} is 
\beq \label{eq: Phi gen}
\wx =\Phi_f(x,\ep)= x + 2\epsilon(I - \ep f'(x))^{-1} f(x),
\eeq
where $f'(x)$ denotes the Jacobi matrix of $f(x)$. As a consequence, each component $x_i$ of $x$ is a rational function,
\begin{equation}
x_i=\frac{p_i(x;\epsilon)}{\Delta(x;\epsilon)},
\end{equation}
with a common denominator
\begin{equation}\label{eq Delta}
\Delta(x;\epsilon)=\det(I - \ep f'(x)).
\end{equation}
Clearly, the degrees of all polynomials $p_i$ and $\Delta$ are equal to $n$.


One has the following expression for the Jacobi matrix of the map $\Phi_f$:
\beq \label{Jac}
d\Phi_f(x)=\frac{\partial\wx}{\partial x}=\big(I-\epsilon f'(x)\big)^{-1}\big(I+\epsilon f'(\wx)\big),
\eeq
so that
\beq \label{Jac det}
\det\big(d\Phi_f(x)\big)=\frac{\Delta(\wx;-\epsilon)}{\Delta(x;\epsilon)}.
\eeq


For our investigations here, the notion of a Hirota-Kimura basis will be relevant. It was introduced and studied in some detail in \cite{PPS1}. We recall here the main facts.

For a given birational map $\Phi: \bbR^n \to \bbR^n$, a set of functions $(\varphi_1,\ldots,\varphi_m)$, linearly independent over $\bbR$, is called a {\it Hirota-Kimura basis}
(HK basis), if for every $x\in\bbR^n$ there exists a non-vanishing vector of coefficients
$c=(c_1,\ldots,c_m)$ such that 
\[
c_1\varphi_1(\Phi^i(x))+\ldots+c_m\varphi_m(\Phi^i(x))=0\quad{\rm for\;\; all} \quad i\in\bbZ.
\]
For a given $x\in\bbR^n$, the vector space consisting of all $c\in\bbR^m$ with
this property, say $K(x)$, is called the {\it null-space} of the basis $(\varphi_1,\ldots,\varphi_m)$ at the point $x$.

The notion of HK-bases is closely related to the notion of integrals, even if they cannot be
immediately translated into one another. For instance, if $\dim K(x)=1$, and $K(x)$ is spanned by the vector $(c_1(x),\ldots,c_m(x))$, then the quotients $c_i(x):c_j(x)$ are integrals of the map $\Phi$.

\section{Novel observations and results}
\label{Sect: novel}

The most complicated cases where integrability of the Kahan-Hirota-Kimura discretization was established in \cite{PPS1, PPS2} are 6-dimensional systems including the Clebsch and the Kirchhoff cases of the motion of a rigid body in an ideal fluid (the first one being linearly isomorphic to the $so(4)$-Euler top), and the Lagrange top. A common feature of these systems is the they are Hamiltonian with respect to linear Lie-Poisson brackets, and completely integrable in the Liouville-Arnold sense (possess four independent integrals in involution, two of them being the Casimir functions of the Lie-Poisson bracket). The discretizations turn out to possess four integrals of motion. These integrals are very complicated. However, for each of these discretizations, one simple {\em quadratic-fractional integral} was found. Integrals which are not simple tend to be extremely complex. Nevertheless, we find a compact representation for some of those complex integrals.

\begin{theorem}\label{th volume}
a) Suppose that there exists a symmetric bilinear expression $\widehat P(x,\wx;\epsilon^2)$ such that for $\wx=\Phi_f(x;\epsilon)$ we have
\begin{equation}\label{bilinear exp P}
\widehat P(x,\wx;\epsilon^2)=\frac{p(x;\epsilon^2)}{\Delta(x;\epsilon)}.
\end{equation}
Then the map $\Phi_f(x;\epsilon)$ has an invariant measure
\begin{equation}\label{eq: inv measure}
\frac{dx_1\wedge\ldots\wedge dx_n}{p(x;\epsilon^2)}.
\end{equation}

b) Suppose that there exists another symmetric  bilinear expression $\widehat Q(x,\wx;\epsilon^2)$ such that for for $\wx=\Phi_f(x;\epsilon)$ we have
\begin{equation}\label{bilinear exp Q}
\widehat Q(x,\wx;\epsilon^2)=\frac{q(x;\epsilon^2)}{\Delta(x;\epsilon)}. 
\end{equation}
Then the map $\Phi_f(x;\epsilon)$ has an integral of motion
\begin{equation}\label{eq: int}
J(x;\epsilon^2)=\frac{\widehat P(x,\wx;\epsilon^2)}{\widehat Q(x,\wx;\epsilon^2)}.
\end{equation}
\end{theorem}
\begin{proof}
Changing in \eqref{bilinear exp P} $\epsilon$ to $-\epsilon$, we see that
\begin{equation}\label{eq: ux}
\widehat P(x,\undertilde{x};\epsilon^2)=\frac{p(x;\epsilon^2)}{\Delta(x;-\epsilon)}.
\end{equation}
Applying the map $x\mapsto \wx$ and taking into account symmetry of $\widehat P$, we find:
$$
\widehat P(x,\wx;\epsilon^2)=\frac{p(\wx;\epsilon^2)}{\Delta(\wx;-\epsilon)}.
$$
Comparing this with \eqref{bilinear exp P}, we arrive at
$$
\frac{\Delta(\wx;-\epsilon)}{\Delta(x;\epsilon)}=\frac{p(\wx;\epsilon^2)}{p(x;\epsilon^2)}.
$$
According to \eqref{Jac det}, this is equivalent to the first statement of the theorem. The second one follows from \eqref{eq: ux}. Indeed, we see that
$$
\frac{\widehat P(x,\undertilde{x};\epsilon^2)}{\widehat Q(x,\undertilde{x};\epsilon^2)}=\frac{p(x;\epsilon^2)}{q(x;\epsilon^2)}=\frac{\widehat P(x,\wx;\epsilon^2)}{\widehat Q(x,\wx;\epsilon^2)}.
$$
This finishes the proof. 
\end{proof}

In general, it is not clear how to find bilinear expressions with the property described in Theorem \ref{th volume}. However, the following observation mysteriously holds true in a big number of cases.

\begin{observation} \label{conjecture 1}
Suppose that the Kahan-Hirota-Kimura discretizatrion of $\dot x=f(x)$ possesses a {\em quadratic-fractional integral} of the form
\begin{equation} \label{simple int}
I(x,\epsilon)=\frac{P(x;\epsilon^2)}{Q(x;\epsilon^2)},
\end{equation}
where $P$ and $Q$ as functions of $x$ are polynomials of degree $\le 2$, while as functions of $\epsilon$ they are polynomials of $\epsilon^2$. 
In many cases, the polarizations $\widehat P$,$\widehat Q$ of the polynomials $P$,$Q$, with $\epsilon^2$ replaced by $-\epsilon^2$, satisfy conditions of Theorem \ref{th volume}. In particular, to the quadratic-fractional integral \eqref{simple int}, there corresponds a further {\em bilinear-fractional} integral
\begin{equation} \label{semisimple int}
J(x,\epsilon)=\frac{\widehat P(x,\wx;-\epsilon^2)}{\widehat Q(x,\wx;-\epsilon^2)},
\end{equation}
while either of the numerators of $\widehat P(x,\wx;-\epsilon^2)$, $\widehat Q(x,\wx;-\epsilon^2)$ serves as a densitiy of an invariant measure.
\end{observation} 

The second observation is related to linear {\em Wronskian relations} with constant coefficients, which turn out to exist for all systems we consider in the present paper. These are relations of the type
\begin{equation}\label{wronskian rel}
\sum_{(i,j)\in \mathcal J} \gamma_{ij}w_{ij}(x)=0, \quad \mathrm{where} \quad w_{ij}(x)=\dot{x}_ix_j-x_i\dot{x}_j,
\end{equation}
satisfied on solutions of $\dot x=f(x)$. Here $\mathcal J\subset\{1,2,\ldots,n\}^2$ is some set, and $\gamma_{ij}$ for $(i,j)\in\mathcal J$ are certain constants. Note that the existence of such a relation can be formulated in a different way, by saying that the Wronskians $w_{ij}(x)=\dot{x}_ix_j-x_i\dot{x}_j=x_jf_i(x)-x_if_j(x)$ with $(i,j)\in\mathcal J$ build a HK basis for the continuous time system $\dot x=f(x)$.
\begin{observation} \label{conjecture 2}
In many cases, to a Wronskian relation \eqref{wronskian rel} satisfied on solutions of $\dot x=f(x)$, the discrete Wronskians 
\begin{equation}\label{discr Wronskian}
W^{(\ell)}_{ij}(x)=x_i^{(\ell)}x_j-x_ix_j^{(\ell)}, \quad (i,j)\in\mathcal J,
\end{equation} 
of all orders $\ell\ge 1$, form a HK basis for the map $\Phi_f$. Here we use the notation $x_i^{(\ell)}=x_i\circ \Phi_f^\ell(x;\epsilon)$ for the components of the iterates of $x$ under the map $\Phi_f$.
\end{observation}

\section{General Clebsch flow}
\label{sect Clebsch gen}

The motion of a heavy top and the motion of a rigid body in an ideal fluid can be described by
the so called {\em Kirchhoff equations} 
\begin{equation}\label{eq: Kirch}
\renewcommand{\arraystretch}{2.2}
\left\{\begin{array}{l} \dot{m}=m\times\dfrac{\partial H}{\partial
m}
        +p\times\dfrac{\partial H}{\partial p}, \\
\dot{p}=p\times\dfrac{\partial H}{\partial m},
\end{array}\right.
\end{equation}
where $H=H(m,p)$ is a quadratic polynomial in $m=(m_1,m_2,m_3)\in\bbR^3$ and
$p=(p_1,p_2,p_3)\in\bbR^3$; here $\times$ denotes vector product
in $\bbR^3$. The physical meaning of $m$ is the total angular
momentum, whereas $p$ represents the total linear momentum of the
system. System (\ref{eq: Kirch}) is Hamiltonian with the Hamilton
function $H(m,p)$, with respect to the Lie-Poisson bracket of $e(3)^*$:
\begin{equation}\label{eq: e3 bracket}
\{m_i,m_j\}=m_k,\qquad \{m_i,p_j\}=p_k,
\end{equation}
where $(i,j,k)$ is a cyclic permutation of (1,2,3) (all other
pairwise Poisson brackets of the coordinate functions are obtained
from these by the skew-symmetry, or otherwise vanish). 
The functions
$$
K_1  =  p_1^2+p_2^2+p_3^2,\quad 
K_2  =  m_1p_1+m_2p_2+m_3p_3
$$
are Casimirs of the bracket (\ref{eq: e3 bracket}), thus 
integrals of motion of an arbitrary Hamiltonian system (\ref{eq: Kirch}). The complete integrability
of Kirchhoff equations is guaranteed by the existence of a fourth
integral of motion, functionally independent of $H,K_1,K_2$ and
in Poisson involution with the Hamiltonian $H$.

Consider a homogeneous quadratic Hamiltonian
\begin{equation}\label{gClebsch Ham}
H = \frac{1}{2}\,\langle m, Am\rangle + \frac{1}{2}\,\langle p, Bp\rangle,
\end{equation}
where $A={\rm{diag}}(a_1,a_2,a_3)$ and $B={\rm{diag}}(b_1,b_2,b_3)$.
The corresponding Kirchhoff equations read
\begin{equation}\label{gcl vector}
\left\{ \begin{array}{l}
\dot{m} =  m\times Am+p\times Bp, \vspace{.1truecm}\\
\dot{p} =  p\times Am,
\end{array} \right.
\end{equation}
or, in components,
\beq \label{gClebsch}
\left\{\begin{array}{l}
\dot{m}_1  =  (a_3-a_2)m_2m_3+(b_3-b_2)p_2p_3,  \vspace{.1truecm}\\
\dot{m}_2  =  (a_1-a_3)m_3m_1+(b_1-b_3)p_3p_1,  \vspace{.1truecm} \\
\dot{m}_3  =  (a_2-a_1)m_1m_2+(b_2-b_1)p_1p_2. \vspace{.1truecm} \\
\dot{p}_1  =  a_3m_3p_2-a_2m_2p_3, \vspace{.1truecm} \\
\dot{p}_2  =  a_1m_1p_3-a_3m_3p_1, \vspace{.1truecm} \\
\dot{p}_3  =  a_2m_2p_1-a_1m_1p_2.
\end{array}\right.
\eeq
An additional (fourth) integral of motion exists, if the following {\em Clebsch condition} is satisfied:
\begin{equation}\label{eq: Clebsch cond}
\frac{b_1-b_2}{a_3}+\frac{b_2-b_3}{a_1}+\frac{b_3-b_1}{a_2}=0.
\end{equation}
Note that condition (\ref{eq: Clebsch cond}) is equivalent to the existence of $\omega_1$, $\omega_2$, $\omega_3$ such that
\begin{equation}\label{eq: Clebsch omega}
\omega_1-\omega_2=\frac{b_1-b_2}{a_3},\quad \omega_2-\omega_3=\frac{b_2-b_3}{a_1},\quad
\omega_3-\omega_1=\frac{b_3-b_1}{a_2}.
\end{equation}
Another way to express condition \eqref{eq: Clebsch cond} is to say that the following three expressions have the same value:
\begin{equation}\label{eq: Clebsch theta}
\frac{b_1-b_2}{a_3(a_1-a_2)}=\frac{b_2-b_3}{a_1(a_2-a_3)}=\frac{b_3-b_1}{a_2(a_3-a_1)}=\frac{1}{\beta}.
\end{equation}
With this notation, and taking into account \eqref{eq: Clebsch omega}, we easily derive that there exists a constant $\alpha$ such that
\begin{equation}
a_i=\alpha+\beta\omega_i, \quad b_i=\alpha \omega_i-\beta\omega_j\omega_k.
\end{equation}
Thus, under condition \eqref{eq: Clebsch cond}, the flow \eqref{gClebsch} is a linear combination of two flows: 
\begin{itemize}
\item the {\em first Clebsch flow} corresponding to $\alpha=1$ and $\beta=0$, so that
\begin{equation}\label{eq: Clebsch 1 ab}
a_1=a_2=a_3=1, \quad
b_1=\omega_1,\quad b_2=\omega_2,\quad
b_3=\omega_3,
\end{equation} 
with the Hamilton function $\frac{1}{2}H_1$, where
\begin{equation}\label{Clebsch H1}
H_1=m_1^2+m_2^2 +m_3^2+\omega_1p_1^2+\omega_2p_2^2+\omega_3p_3^2,
\end{equation}
\item and the {\em second Clebsch flow} corresponding to $\alpha=0$ and $\beta=1$, so that
\begin{equation}\label{eq: Clebsch 2 ab}
\left\{\begin{array}{l}
a_1=\omega_1,\quad a_2=\omega_2,\quad a_3=\omega_3,\\
b_1=-\omega_2\omega_3,\quad b_2=-\omega_3\omega_1,\quad
b_3=-\omega_1\omega_2.
\end{array}\right.
\end{equation}
with the Hamilton function $\frac{1}{2}H_2$, where
\begin{equation}\label{Clebsch H2}
H_2 = \omega_1m_1^2+\omega_2m_2^2+\omega_3m_3^2 -\omega_2\omega_3p_1^2-\omega_3\omega_1p_2^2-\omega_1\omega_2p_3^2.
\end{equation}
\end{itemize}
The key statement is:
\begin{theorem}
Hamilton functions \eqref{Clebsch H1} and \eqref{Clebsch H2} are in involution, so that the first and the second Clebsch flows commute. Thus, under condition \eqref{eq: Clebsch cond}, the flow \eqref{gClebsch} is completely integrable.
\end{theorem}

As a matter of fact, Clebsch condition \eqref{eq: Clebsch cond} admits a different characterization.
\begin{theorem}
Condition \eqref{eq: Clebsch cond} is equivalent to the existence of a Wronskian relation  
\begin{equation}\label{eq: Clebsch Wronski}
    A_1(\dot{m}_1p_1-m_1\dot{p}_1)+A_2(\dot{m}_2p_2-m_2\dot{p}_2)+A_3(\dot{m}_3p_3-m_3\dot{p}_3)=0
\end{equation}
with constant coefficients $A_i$, satisfied on solutions of \eqref{gClebsch}. In this case, coefficients $A_i$ are given, up to a common factor, by
\begin{equation}\label{syst A sol}
A_1=\frac{1}{a_2}+\frac{1}{a_3}-\frac{1}{a_1}, \quad A_2=\frac{1}{a_3}+\frac{1}{a_1}-\frac{1}{a_2}, \quad A_3=\frac{1}{a_1}+\frac{1}{a_2}-\frac{1}{a_3}.
\end{equation}
\end{theorem}
\begin{proof} Wronskian relation \eqref{eq: Clebsch Wronski} is satisfied by virtue of equations of motion \eqref{gClebsch}, if and only if coefficients $A_i$ satisfy
\begin{equation}\label{syst A}
\left\{\begin{array}{l} A_1(a_3-a_2)+A_2a_3-A_3a_2=0,\vspace{.1truecm} \\ -A_1a_3+A_2(a_1-a_3)+A_3a_1=0, \vspace{.1truecm} \\ A_1a_2-A_2a_1+A_3(a_2-a_1)=0, \end{array}\right. 
\end{equation}
and 
\begin{equation}\label{eq A}
A_1(b_3-b_2)+A_2(b_1-b_3)+A_3(b_2-b_1)=0.
\end{equation}
System \eqref{syst A} is equivalent to
$$
(A_1+A_2)a_3=(A_3+A_1)a_2=(A_2+A_3)a_1,
$$
so that, up to an inessential common factor, we find:
 \begin{equation}\label{syst A sol aux}
A_1+A_2=\frac{2}{a_3}, \quad A_2+A_3=\frac{2}{a_1}, \quad A_3+A_1=\frac{2}{a_2},
\end{equation}
leading to \eqref{syst A sol}. With these values of $A_i$, equation \eqref{eq A} is equivalent to Clebsch condition \eqref{eq: Clebsch cond}:
\begin{align*}
& A_1(b_3-b_2)+A_2(b_1-b_3)+A_3(b_2-b_1) \\
& \qquad =b_1(A_2-A_3)+b_2(A_3-A_1)+b_3(A_1-A_2) \\
 & \qquad = b_1\left(\frac{1}{a_3}-\frac{1}{a_2}\right)+b_2\left(\frac{1}{a_1}-\frac{1}{a_3}\right)+b_3\left(\frac{1}{a_2}-\frac{1}{a_1}\right) \\
 & \qquad =  \frac{b_1-b_2}{a_3}+\frac{b_2-b_3}{a_1}+\frac{b_3-b_1}{a_2} =0,
 \end{align*}
 which concludes the proof.
\end{proof}

In what follows, we assume that Clebsch condition \eqref{eq: Clebsch theta} is satisfied with $\beta\neq 0$, and we define the constants $A_i$ as in \eqref{syst A sol}.
\smallskip

Applying the Kahan-Hirota-Kimura scheme to the general flow (\ref{gClebsch}) of Clebsch
system, we arrive at the following discretization:
\beq\label{eq:dC2}
 \left\{ \begin{array}{l}
\widetilde{m}_1-m_1 =
\epsilon(a_3-a_2)(\widetilde{m}_2m_3+m_2\widetilde{m}_3)+
\epsilon(b_3-b_2)(\widetilde{p}_2p_3+p_2\widetilde{p}_3),
\vspace{.1truecm} \\
\widetilde{m}_2-m_2 =
\epsilon(a_1-a_3)(\widetilde{m}_3m_1+m_3\widetilde{m}_1)+
\epsilon(b_1-b_3)(\widetilde{p}_3p_1+p_3\widetilde{p}_1),
\vspace{.1truecm}\\
\widetilde{m}_3-m_3 =
\epsilon(a_2-a_1)(\widetilde{m}_1m_2+m_1\widetilde{m}_2)+
\epsilon(b_2-b_1)(\widetilde{p}_1p_2+p_1\widetilde{p}_2),
\vspace{.1truecm} \\
\widetilde{p}_1-p_1 = \epsilon
a_3(\widetilde{m}_3p_2+m_3\widetilde{p}_2)-\epsilon
a_2(\widetilde{m}_2 p_3+m_2\widetilde{p}_3),
\vspace{.1truecm}\\
\widetilde{p}_2-p_2= \epsilon
a_1(\widetilde{m}_1p_3+m_1\widetilde{p}_3)-\epsilon
a_3(\widetilde{m}_3p_1+m_3\widetilde{p}_1),
\vspace{.1truecm} \\
\widetilde{p}_3-p_3 = \epsilon
a_2(\widetilde{m}_2p_1+m_2\widetilde{p}_1)-\epsilon
a_1(\widetilde{m}_1 p_2+m_1\widetilde{p}_2).
\end{array}\right.
\end{equation}
Linear system (\ref{eq:dC2}) defines a birational map $\Phi_f:\bbR^6 \rightarrow \bbR^6$, $(m,p)\mapsto (\wm,\wip)$.

A simple (quadratic-fractional) integral of $\Phi_f$ and a first order Wronskian HK basis for this map were found in our previous work.

\begin{theorem}\label{Th: dC2 basis 1}  {\bf (Quadratic-fractional integral, \cite{PPS1})}
The function
\begin{equation}\label{eq: dC2 C1} 
I_0(m,p,\epsilon)=\frac{A_1a_2a_3g_1+A_2a_3a_1g_2+A_3a_1a_2g_3}
{1+\epsilon^2\dfrac{a_1a_2a_3}{\beta}(g_1+g_2+g_3)},
\end{equation}
where
\begin{equation}\label{eq: gi}
g_i(m,p)=p_i^2+\frac{\beta a_i}{a_ja_k}m_i^2,
\end{equation}
is an integral of motion of the map $\Phi_f$. The set 
\begin{equation}\label{eq: dC2 basis 1}
\Psi_0=(g_1,g_2,g_3,1)
\end{equation}
is a HK basis for the map $\Phi_f$, with a one-dimensional null space
\begin{equation}
K_{\Psi_0}(m,p)=[c_1a_2a_3:c_2a_3a_1:c_3a_1a_2:-c_0],
\end{equation}
where
\begin{eqnarray}
c_1 & = & A_1+\epsilon^2 A_3a_1(b_1-b_2)g_2+\epsilon^2 A_2a_1(b_1-b_3)g_3, \label{Clebsch2 c1}\\
c_2 & = & A_2+\epsilon^2 A_1a_2(b_2-b_3)g_3+\epsilon^2 A_3a_2(b_2-b_1)g_1,  \label{Clebsch2 c2}\\
c_3 & = & A_3+\epsilon^2 A_2a_3(b_3-b_1)g_1+\epsilon^2 A_1a_3(b_3-b_2)g_2,  \label{Clebsch2 c3} \\
c_0 & = & A_1a_2a_3g_1+A_2a_3a_1g_2+A_3a_1a_2g_3. \label{Clebsch2 c0}
\end{eqnarray}
\end{theorem}
\begin{theorem} \label{Th: Clebsch2 W1 basis} {\bf (Discrete Wronskians HK basis, \cite{PPS2})}
Functions $W^{(1)}_i(m,p)=\wm_ip_i-m_i\wip_i$, $i=1,2,3$, form a HK basis for the map $\Phi_f$, with a one-dimensional null space spanned by
$[c_1:c_2:c_3]$, with the functions $c_i$ given in \eqref{Clebsch2 c1}--\eqref{Clebsch2 c3}. In other words, on orbits of the map $\Phi_f$ there holds: 
\begin{equation}
\sum_{i=1}^3 c_i(\wm_ip_i-m_i\wip_i) =0.
\end{equation}
\end{theorem}

Novel results, exemplifying Observations \ref{conjecture 1} and \ref{conjecture 2}, are as follows.

\begin{theorem}\label{Th: dC2 basis 2} {\bf (Bilinear-fractional integral)}
The function
\begin{equation}\label{eq: dC2 C2} 
J_0(m,p,\epsilon)=\frac{A_1a_2a_3G_1+A_2a_3a_1G_2+A_3a_1a_2G_3}
{1-\epsilon^2\dfrac{a_1a_2a_3}{\beta}(G_1+G_2+G_3)},
\end{equation}
where
\begin{equation}\label{eq: Gi}
G_i(m,p)=p_i\wip_i+\frac{\beta a_i}{a_ja_k}m_i\wm_i,
\end{equation}
is an integral of motion of the map $\Phi_f$. The set 
\begin{equation}\label{eq: dClebsch2 basis 2}
\Psi_1=(G_1,G_2,G_3,1)
\end{equation}
is a HK basis for the map $\Phi_f$, with a one-dimensional null space
\begin{equation}
K_{\Psi_1}(m,p)=[C_1a_2a_3:C_2a_3a_1:C_3a_1a_2:-C_0],
\end{equation}
where
\begin{eqnarray}
C_1 & = & A_1-\epsilon^2 A_3a_1(b_1-b_2)G_2-\epsilon^2 A_2a_1(b_1-b_3)G_3, \label{Clebsch2 C1}\\
C_2 & = & A_2-\epsilon^2 A_1a_2(b_2-b_3)G_3-\epsilon^2 A_3a_2(b_2-b_1)G_1,  \label{Clebsch2 C2}\\ 
C_3 & = & A_3-\epsilon^2 A_2a_3(b_3-b_1)G_1-\epsilon^2 A_1a_3(b_3-b_2)G_2,  \label{Clebsch2 C3} \\
C_0 & = & A_1a_2a_3G_1+A_2a_3a_1G_2+A_3a_1a_2G_3. \label{Clebsch2 C0}
\end{eqnarray}
\end{theorem}

Actually, the numerator and the denominator of the fraction \eqref{eq: dC2 C2} satisfy conditions of Theorem \ref{th volume}.

\begin{corollary}\label{Th: Clebsch2 conserved density} {\bf (Density of an invariant measure)}
The map $\Phi_f(x;\epsilon)$ has an invariant measure
$$
\frac{dm_1\wedge dm_2\wedge dm_3 \wedge dp_1\wedge dp_2\wedge dp_3}{\phi(m,p;\epsilon)},
$$
where $\phi(m,p;\epsilon)$ can be taken as the numerator of either of the functions $C_i$, $i=0,1,2,3$.
\end{corollary}

\begin{theorem}  \label{Th: Clebsch2 W2 basis} {\bf (Second order Wronskians HK basis)}
The functions $$W^{(2)}_i(m,p)=\widetilde{\wm}_ip_i-m_i\widetilde{\wip}_i, \quad i=1,2,3,$$ form a HK basis for the map $\Phi_f$, with a one-dimensional null space spanned by
$[C_1:C_2:C_3]$, with the functions $C_i$ given in \eqref{Clebsch2 C1}--\eqref{Clebsch2 C3}. In other words, on orbits of the map $\Phi_f$ there holds: 
\begin{equation}
\sum_{i=1}^3 C_i(\wwm_ip_i-m_i\wwp_i) =0.
\end{equation}
\end{theorem}

\begin{theorem} \label{Clebsch higher Wronskians} {\bf (Higher order Wronskians HK bases)}
\begin{enumerate}
\item The functions $W^{(3)}_i(m,p)=\widetilde{\widetilde {\widetilde{m}}}_ip_i-m_i\widetilde{\widetilde {\widetilde{p}}}_i$, $i=1,2,3$, form a HK basis for the map $\Phi_f$, with a one-dimensional null space.
On orbits of the map $\Phi_f$ there holds
$$
\sum_{i=1}^3 D_i(\widetilde{\widetilde {\widetilde{m}}}_ip_i-m_i\widetilde{\widetilde {\widetilde{p}}}_i) =0.
$$
The two integrals of motion $J_1=D_1/D_3$ and $J_2=D_2/D_3$ are functionally independent.
\smallskip

\item The functions $W^{(4)}_i(m,p)=\widetilde{\widetilde{\widetilde {\widetilde{m}}}}_ip_i-m_i\widetilde{\widetilde{\widetilde {\widetilde{p}}}}_i$, $i=1,2,3$, form a HK basis for the map $\Phi_f$, with a one-dimensional null space.
On orbits of the map $\Phi_f$ there holds
$$
\sum_{i=1}^3 E_i(\widetilde{\widetilde{\widetilde {\widetilde{m}}}}_ip_i-m_i\widetilde{\widetilde{\widetilde {\widetilde{p}}}}_i) =0.
$$
The two integrals of motion $J_3=E_1/E_3$ and $J_4=E_2/E_3$  are functionally independent.
\smallskip

\item Among the integrals of motion $I_0,J_0,\dots,J_4$ four are functionally independent. In particular, each of the sets
$\{I_0,J_0,J_1,J_2\}$, $\{I_0,J_0,J_3,J_4\}$, $\{J_1,J_2,J_3,J_4\}$ consists of four independent integrals of motion.
\end{enumerate}
\end{theorem}

\section{The first Clebsch flow}
\label{sect Clebsch 1}

We collect here the formulas for the first Clebsch flow, since this will be our main example. Equations of motion:
\beq \label{Cl_1}
\left\{\begin{array}{l}
\dot{m}_1  =  (\omega_3-\omega_2)p_2p_3,  \vspace{.1truecm}  \\
\dot{m}_2  = (\omega_1-\omega_3)p_3p_1,  \vspace{.1truecm}\\
\dot{m}_3  =  (\omega_2-\omega_1)p_1p_2,  \vspace{.1truecm} \\
\dot{p}_1  =  m_3 p_2 - m_2 p_3,   \vspace{.1truecm}\\
\dot{p}_2  =  m_1 p_3 - m_3 p_1,  \vspace{.1truecm}\\
\dot{p}_3  =  m_2 p_1 - m_1 p_2.  
\end{array}\right.
\eeq
Wronskian relation satisfied on solutions of \eqref{Cl_1} is:
\begin{equation}
\label{eq: Clebsch W}
(\dot{m}_1p_1-m_1\dot{p}_1)+(\dot{m}_2p_2-m_2\dot{p}_2)+(\dot{m}_3p_3-m_3\dot{p}_3)=0.
\end{equation}

Applying the Kahan-Hirota-Kimura scheme to the first Clebsch flow (\ref{Cl_1}), we arrive at the following discretization:
\beq \label{eq: dC}
\left\{\begin{array}{l}
\widetilde{m}_1-m_1  =  \epsilon(\omega_3-\omega_2)
(\widetilde{p}_2p_3+p_2\widetilde{p}_3),         \vspace{.1truecm} \\
\widetilde{m}_2-m_2  =  \epsilon(\omega_1-\omega_3)
(\widetilde{p}_3p_1+p_3\widetilde{p}_1),         \vspace{.1truecm} \\
\widetilde{m}_3-m_3  =  \epsilon(\omega_2-\omega_1)
(\widetilde{p}_1p_2+p_1\widetilde{p}_2),         \vspace{.1truecm} \\
\widetilde{p}_1-p_1  = 
\epsilon(\widetilde{m}_3p_2+m_3\widetilde{p}_2)-
\epsilon(\widetilde{m}_2p_3+m_2\widetilde{p}_3), \vspace{.1truecm} \\
\widetilde{p}_2-p_2  = 
\epsilon(\widetilde{m}_1p_3+m_1\widetilde{p}_3)-
\epsilon(\widetilde{m}_3p_1+m_3\widetilde{p}_1),  \vspace{.1truecm} \\
\widetilde{p}_3-p_3  = 
\epsilon(\widetilde{m}_2p_1+m_2\widetilde{p}_1)-
\epsilon(\widetilde{m}_1p_2+m_1\widetilde{p}_2). 
\end{array}\right.
\eeq
Linear system (\ref{eq: dC}) defines a birational map $\Phi_f:\bbR^6 \rightarrow \bbR^6$,
$(m,p)\mapsto (\widetilde{m}, \widetilde{p})$.

\begin{theorem}\label{Th: dC basis 1}  {\bf (Quadratic-fractional integral, \cite{PPS1})}
The function
\begin{equation}\label{eq: dC C1} 
I_0(m,p,\epsilon)=\frac{p_1^2+p_2^2+p_3^2}
{1-\epsilon^2(\omega_1p_1^2+\omega_2p_2^2+\omega_3p_3^2)}
\end{equation}
is an integral of motion of the map $\Phi_f:(m,p)\mapsto(\wm,\wip)$. The set 
\begin{equation}\label{eq: dClebsch basis 1}
\Psi_0=(p_1^2,p_2^2,p_3^2,1)
\end{equation}
is a HK-basis for the map $\Phi_f$, with  a one-dimensional null-space 
\begin{equation}
K_{\Psi_0}(m,p)=[c_1:c_2:c_3:-c_0],
\end{equation}
where
\begin{eqnarray}
c_1 & = & 1+\epsilon^2(\omega_1-\omega_2)p_2^2+\epsilon^2(\omega_1-\omega_3)p_3^2, \label{Clebsch1 c1}\\
c_2 & = & 1+\epsilon^2(\omega_2-\omega_1)p_1^2+\epsilon^2(\omega_2-\omega_3)p_3^2, \label{Clebsch1 c2}\\
c_3 & = & 1+\epsilon^2(\omega_3-\omega_1)p_1^2+\epsilon^2(\omega_3-\omega_2)p_2^2, \label{Clebsch1 c3} \\
c_0 & = & p_1^2+p_2^2+p_3^2. \label{Clebsch1 c0}
\end{eqnarray}
Equivalently:
\begin{eqnarray}\label{eq: dC ans1 res}
K_{\Psi_0}(m,p)  & = & [\ 1+\epsilon^2\omega_1 I_0:1+\epsilon^2\omega_2 I_0:
1+\epsilon^2\omega_3 I_0: -I_0\ ] \\
& = & \left[\,\frac{1}{I_0}+\epsilon^2\omega_1:\frac{1}{I_0}+\epsilon^2\omega_2:
\frac{1}{I_0}+\epsilon^2\omega_3:-1\,\right].
\end{eqnarray}
\end{theorem}

\begin{theorem} \label{Th: Clebsch1 W1 basis} {\bf (Discrete Wronskians HK basis, \cite{PPS2})}
The functions $W^{(1)}_i(m,p)=\wm_ip_i-m_i\wip_i$, $i=1,2,3$, form a HK basis for the map $\Phi_f$, with a one-dimensional null space spanned by
$[c_1:c_2:c_3]$, with the functions $c_i$ given in \eqref{Clebsch1 c1}--\eqref{Clebsch1 c3}. In other words, on orbits of the map $\Phi_f$ there holds: 
\begin{equation}
\sum_{i=1}^3 c_i(\wm_ip_i-m_i\wip_i) =0.
\end{equation}
\end{theorem}

Novel results, illustrating Observations \ref{conjecture 1} and \ref{conjecture 2}, are as follows.

\begin{theorem}\label{Th: dC basis 2} {\bf (Bilinear-fractional integral)}
The function
\begin{equation}\label{eq: dC C2} 
J_0(m,p,\epsilon)=\frac{p_1\wip_1+p_2\wip_2+p_3\wip_3}
{1+\epsilon^2(\omega_1p_1\wip_1+\omega_2p_2\wip_2+\omega_3p_3\wip_3)}
\end{equation}
is an integral of motion of the map $\Phi_f$. The set 
\begin{equation}\label{eq: dClebsch basis 2}
\Psi_1=(p_1\wip_1,p_2\wip_2,p_3\wip_3,1)
\end{equation}
is a HK-basis for the map $\Phi_f$, with a one-dimensional null space
\begin{equation}
K_{\Psi_1}(m,p)=[C_1:C_2:C_3:-C_0],
\end{equation}
where
\begin{eqnarray}
C_1 & = & 1+\epsilon^2(\omega_2-\omega_1)p_2\wip_2+\epsilon^2(\omega_3-\omega_1)p_3\wip_3, \label{Clebsch1 C1}\\
C_2 & = & 1+\epsilon^2(\omega_1-\omega_2)p_1\wip_1+\epsilon^2(\omega_3-\omega_2)p_3\wip_3, \label{Clebsch1 C2}\\
C_3 & = & 1+\epsilon^2(\omega_1-\omega_3)p_1\wip_1+\epsilon^2(\omega_2-\omega_3)p_2\wip_2, \label{Clebsch1 C3}\\
C_0 & = & p_1\wip_1+p_2\wip_2+p_3\wip_3. \label{Clebsch1 C0}
\end{eqnarray}
Equivalently:
\begin{eqnarray}\label{eq: dC ans2 res}
K_{\Psi_1}(m,p)  & = & [\ 1-\epsilon^2\omega_1 J_0:1-\epsilon^2\omega_2 J_0:
1-\epsilon^2\omega_3 J_0: -J_0\ ] \\
& = & \left[\,\frac{1}{J_0}-\epsilon^2\omega_1:\frac{1}{J_0}-\epsilon^2\omega_2:
\frac{1}{J_0}-\epsilon^2\omega_3:-1\,\right].
\end{eqnarray}
\end{theorem}

\begin{remark} One can show that the numerators of the expressions $C_i$, $i=0,\ldots,3$, are irreducible polynomials of $m,p$ depending on $\epsilon^2$ rather than on $\epsilon$, thus they satisfy Observation \ref{conjecture 1}.

\end{remark}

\begin{corollary}\label{Th: Clebsch1 conserved density} {\bf (Density of an invariant measure)}
The map $\Phi_f(x;\epsilon)$ has an invariant measure
$$
\frac{dm_1\wedge dm_2\wedge dm_3 \wedge dp_1\wedge dp_2\wedge dp_3}{\phi(m,p;\epsilon)},
$$
where for $\phi(m,p;\epsilon)$ one can take the numerator of the function $p_1\wip_1+p_2\wip_2+p_3\wip_3$, or the numerator of the function $1+\epsilon^2(\omega_1p_1\wip_1+\omega_2p_2\wip_2+\omega_3p_3\wip_3)$ (the quotient of both densities is an integral of motion $J_0$).
\end{corollary}

\begin{theorem}  \label{Th: Clebsch1 W2 basis} {\bf (Second order Wronskians HK basis)}
The functions $$W^{(2)}_i(m,p)=\widetilde{\wm}_ip_i-m_i\widetilde{\wip}_i, \quad i=1,2,3,$$ form a HK basis for the map $\Phi_f$, with a one-dimensional null space spanned by
$[C_1:C_2:C_3]$, with the functions $C_i$ given in \eqref{Clebsch1 C1}--\eqref{Clebsch1 C3}. In other words, on orbits of the map $\Phi_f$ there holds: 
\begin{equation}
\sum_{i=1}^3 C_i(\wwm_ip_i-m_i\wwp_i) =0.
\end{equation}
\end{theorem}

Finally, there holds a theorem which reads literally as Theorem \ref{Clebsch higher Wronskians} on higher order Wronskians HK bases.
\smallskip

We now turn to the proofs of the above results.
\smallskip

\noindent
{\bf Proof of Theorem \ref{Th: dC basis 1}.}
We show that the function $I_0$ in (\ref{eq: dC C1}) is an integral of motion, i.e.,
that
\[
\frac{p_1^2+p_2^2+p_3^2}
{1-\epsilon^2(\omega_1p_1^2+\omega_2p_2^2+\omega_3p_3^2)}=
\frac{\widetilde{p}_1^2+\widetilde{p}_2^2+\widetilde{p}_3^2}
{1-\epsilon^2(\omega_1\widetilde{p}_1^2+\omega_2\widetilde{p}_2^2+\omega_3
\widetilde{p}_3^2)}.
\]
This is equivalent to
\begin{eqnarray*}
\lefteqn{\widetilde{p}_1^{2}-p_1^2+\widetilde{p}_2^2-p_2^2+\widetilde{p}_3^2-p_3^2}\\
&=&\epsilon^2\left[(\omega_2-\omega_1)(\widetilde{p}_1^2p_2^2-\widetilde{p}_2^2p_1^2)
+(\omega_3-\omega_2)(\widetilde{p}_2^2p_3^2-\widetilde{p}_3^2p_2^2)
+(\omega_1-\omega_3)(\widetilde{p}_3^2p_1^2-\widetilde{p}_1^2p_3^2)\right].
\end{eqnarray*}
On the left-hand side of this equation we replace
$\widetilde{p}_i-p_i$ through the expressions from the last three
equations of motion (\ref{eq: dC}), on the right-hand side we
replace
$\epsilon(\omega_k-\omega_j)(\widetilde{p}_jp_k+p_j\widetilde{p}_k)$
by $\widetilde{m}_i-m_i$, according to the first three equations
of motion (\ref{eq: dC}). This brings the equation we want to
prove into the form
\begin{eqnarray*}
\lefteqn{(\widetilde{p}_1+p_1)(\widetilde{m}_3p_2+m_3\widetilde{p}_2-
\widetilde{m}_2p_3-m_2\widetilde{p}_3)\;+} & & \\
\lefteqn{(\widetilde{p}_2+p_2)(\widetilde{m}_1p_3+m_1\widetilde{p}_3-
\widetilde{m}_3p_1-m_3\widetilde{p}_1)\;+} & & \\
\lefteqn{(\widetilde{p}_3+p_3)(\widetilde{m}_2p_1+m_2\widetilde{p}_1-
\widetilde{m}_1p_2-m_1\widetilde{p}_2)\;=} & & \\
&&=(\widetilde{p}_1p_2-p_1\widetilde{p}_2)(\widetilde{m}_3-m_3)+
 (\widetilde{p}_2p_3-p_2\widetilde{p}_3)(\widetilde{m}_1-m_1)+
 (\widetilde{p}_3p_1-p_3\widetilde{p}_1)(\widetilde{m}_2-m_2).
\end{eqnarray*}
But this is an algebraic identity in twelve variables $m_k,p_k,\widetilde{m}_k,\widetilde{p}_k$. 
\hfill $\blacksquare$

\smallskip
\noindent

{\bf Proof of Theorem \ref{Th: dC basis 2}.} 
We show that the function $J_0$ in (\ref{eq: dC C2}) is an integral of motion, i.e.,
that
\[
\frac{p_1\up_1+p_2\up_2+p_3\up_3}
{1+\epsilon^2(\omega_1p_1\up_1+\omega_2p_2\up_2+\omega_3p_3\up_3)}=
\frac{\wip_1p_1+\wip_2p_2+\wip_3p_3}
{1+\epsilon^2(\omega_1\wip_1p_1+\omega_2\wip_2p_2+\omega_3\wip_3p_3)}.
\]
This is equivalent to
\begin{eqnarray*}
\lefteqn{p_1(\wip_1-\up_1)+p_2(\wip_2-\up_2)+p_3(\wip_3-\up_3)}\\
&=&\epsilon^2\left[(\omega_1-\omega_2)(\wip_1p_1p_2\up_2-\wip_2p_2p_1\up_1)
+(\omega_2-\omega_3)(\wip_2p_2p_3\up_3-\wip_3p_3p_2\up_2)\right.\\
& & \left.+(\omega_3-\omega_1)(\wip_3p_3p_1\up_1-\wip_1p_1p_3\up_3)\right].
\end{eqnarray*}
On the left-hand side of this equation we use
$$
2p_i(\wip_i-\untilde{p}_i)=(\wip_i+p_i)(p_i-\untilde{p}_i)+(\wip_i-p_i)(p_i+\up_i),
$$
and then replace $\wip_i-p_i$ and $p_i-\up_i$ through the expressions from the last three equations of motion (\ref{eq: dC}).
On the right-hand side we use
$$
2(\wip_jp_jp_k\up_k-\wip_kp_kp_j\up_j)=(\wip_jp_k+\wip_kp_j)(p_j\up_k-p_k\up_j)+(\wip_jp_k-\wip_kp_j)(p_j\up_k+p_k\up_j),
$$
and then replace
$\epsilon(\omega_k-\omega_j)(\widetilde{p}_jp_k+p_j\widetilde{p}_k)$
by $\widetilde{m}_i-m_i$, and $\epsilon(\omega_k-\omega_j)(p_j\up_k+p_k\up_j)$
by $m_i-\um_i$, according to the first three equations
of motion (\ref{eq: dC}). This brings the equation we want to
prove into the form
\begin{eqnarray*}
&&(p_1+\up_1)(\wm_3p_2+m_3\wip_2-\wm_2p_3-m_2\wip_3)+(\wip_1+p_1)(m_3\up_2+\um_3p_2-m_2\up_3-\um_2p_3)\;+  \\
&&(p_2+\up_2)(\wm_1p_3+m_1\wip_3-\wm_3p_1-m_3\wip_1)+(\wip_2+p_2)(m_1\up_3+\um_1p_3-m_3\up_1-\um_3p_1)\;+  \\
&&(p_3+\up_3)(\wm_2p_1+m_2\wip_1-\wm_1p_2-m_1\wip_2)+(\wip_3+p_3)(m_2\up_1+\um_2p_1-m_1\up_2-\um_1p_2)  \\
&&\qquad=  (p_2\up_1-\up_2p_1)(\wm_3-m_3)+(\wip_2p_1-p_2\wip_1)(m_3-\um_3) \\
&&\qquad\quad + (p_3\up_2-\up_3p_2)(\wm_1-m_1)+(\wip_3p_2-p_3\wip_2)(m_1-\um_1)\\
&&\qquad\quad+ (p_1\up_3-\up_1p_3)(\wm_2-m_2)+(\wip_1p_3-p_1\wip_3)(m_2-\um_2).
\end{eqnarray*}
On the left-hand side there are many cancellations, so that we get:
\begin{eqnarray*}
&&p_1(m_3\wip_2-m_2\wip_3)+p_1(m_3\up_2-m_2\up_3)+\up_1(\wm_3p_2-\wm_2p_3)+\wip_1(\um_3p_2-\um_2p_3)+  \\
&&p_2(m_1\wip_3-m_3\wip_1)+p_2(m_1\up_3-m_3\up_1)+\up_2(\wm_1p_3-\wm_3p_1)+\wip_2(\um_1p_3-\um_3p_1)+  \\
&&p_3(m_2\wip_1-m_1\wip_2)+p_3(m_2\up_1-m_1\up_2)+\up_3(\wm_2p_1-\wm_1p_2)+\wip_3(\um_2p_1-\um_1p_2)  \\
&&\qquad=  (p_2\up_1-\up_2p_1)(\wm_3-m_3)+(\wip_2p_1-p_2\wip_1)(m_3-\um_3) \\
&&\qquad\quad + (p_3\up_2-\up_3p_2)(\wm_1-m_1)+(\wip_3p_2-p_3\wip_2)(m_1-\um_1)\\
&&\qquad\quad+ (p_1\up_3-\up_1p_3)(\wm_2-m_2)+(\wip_1p_3-p_1\wip_3)(m_2-\um_2).
\end{eqnarray*}
But this is an algebraic identity in the variables $m_k,p_k,\widetilde{m}_k,\widetilde{p}_k,\um_k,\up_k$. 
\hfill $\blacksquare$
\smallskip
\noindent

{\bf Proof of Theorem \ref{Th: Clebsch1 W1 basis}} is based on the following four identities which hold on orbits of the map $\Phi_f$:
\begin{eqnarray}
\sum_{i=1}^3 c_i\wm_ip_i  & = &  \sum_{i=1}^3 C_im_ip_i,  \label{eq: 1}\\
\sum_{i=1}^3 c_im_i\wip_i & = & \sum_{i=1}^3 C_im_ip_i, \label{eq: 2} \\
\sum_{i=1}^3 \widetilde{c}_im_i\wip_i & = & \sum_{i=1}^3 C_i\wm_i\wip_i, \label{eq: 3}\\
\sum_{i=1}^3 \widetilde{c}_i\wm_ip_i & =  &\sum_{i=1}^3 C_i\wm_i\wip_i. \label{eq: 4}
\end{eqnarray}
Indeed, from these relations there follows immediately:
$$
\sum_{i=1}^3 c_i(\wm_ip_i -m_i\wip_i)=\sum_{i=1}^3 \widetilde{c}_i(\wm_ip_i -m_i\wip_i)=0,
$$
which proves the theorem.
\hfill $\blacksquare$

{\bf Proof of formula \eqref{eq: 1}.}
Using the first three equations of motion \eqref{eq: dC}, we compute:
\begin{eqnarray*}
\lefteqn{\sum_{i=1}^3 (\wm_i-m_i)p_i}\\ 
& = &\epsilon(\omega_3-\omega_2)(p_1p_2\wip_3+p_1p_3\wip_2)+\epsilon(\omega_1-\omega_3)(p_2p_3\wip_1+p_2p_1\wip_3)\\
 & & +\epsilon(\omega_2-\omega_1)(p_3p_1\wip_2+p_3p_2\wip_1) \\ \\
& = & \epsilon (\omega_1-\omega_2)p_1p_2\wip_3+\epsilon (\omega_2-\omega_3)p_2p_3\wip_1+\epsilon (\omega_3-\omega_1)p_3p_1\wip_2 \\ \\
& = & \epsilon (\omega_1-\omega_2)p_1p_2(\wip_3-p_3)+\epsilon (\omega_2-\omega_3)p_2p_3(\wip_1-p_1)+\epsilon (\omega_3-\omega_1)p_3p_1(\wip_2-p_2).
\end{eqnarray*}
Using the last three equations of motion \eqref{eq: dC}, we continue the computation:
\begin{eqnarray*}
& = & \epsilon^2 (\omega_1-\omega_2)p_1p_2\big((\wm_2p_1+m_2\wip_1)-(\wm_1p_2+m_1\wip_2)\big)\\
 & & +\epsilon^2 (\omega_2-\omega_3)p_2p_3\big((\wm_3p_2+m_3\wip_2)-(\wm_2p_3+m_2\wip_3)\big)\\
 & & +\epsilon^2 (\omega_3-\omega_1)p_3p_1\big((\wm_1p_3+m_1\wip_3)-(\wm_3p_1+m_3\wip_1)\big)\\ \\
 & = &\; \; \big(\epsilon^2(\omega_2-\omega_1)p_2^2+\epsilon^2(\omega_3-\omega_1)p_3^2\big)\wm_1p_1
          +\big(\epsilon^2(\omega_2-\omega_1)\wip_2p_2+\epsilon^2(\omega_3-\omega_1)\wip_3p_3\big)m_1p_1 \\
 & &   +\big(\epsilon^2(\omega_3-\omega_2)p_3^2+\epsilon^2(\omega_1-\omega_2)p_1^2\big)\wm_2p_2
          +\big(\epsilon^2(\omega_3-\omega_2)\wip_3p_3+\epsilon^2(\omega_1-\omega_2)\wip_1p_1\big)m_2p_2\\
 & & +\big(\epsilon^2(\omega_1-\omega_3)p_1^2+\epsilon^2(\omega_2-\omega_3)p_2^2\big)\wm_3p_3
        +\big(\epsilon^2(\omega_1-\omega_3)\wip_1p_1+\epsilon^2(\omega_2-\omega_3)\wip_2p_2\big)m_3p_3.
\end{eqnarray*}
This proves \eqref{eq: 1}.
We remark that this relation is the discrete time analog of $\sum_{i=1}^3\dot{m}_ip_i=0$; the crucial point in this analogy is that $\sum_{i=1}^3(\wm_i-m_i)p_i$ turns out to be of order $\epsilon^2$. 
\hfill $\blacksquare$
\smallskip

{\bf Proof of formula \eqref{eq: 2}.}
Using the first three equations of motion \eqref{eq: dC}, we compute:
\begin{eqnarray*}
\lefteqn{\sum_{i=1}^3 m_i(\wip_i-p_i)}\\ 
& = & \epsilon m_1\big((\wm_3p_2+m_3\wip_2)-(\wm_2p_3+m_2\wip_3)\big)
         +\epsilon m_2\big((\wm_1p_3+m_1\wip_3)-(\wm_3p_1+m_3\wip_1)\big)\\
 & & +\epsilon m_3\big((\wm_2p_1+m_2\wip_1)-(\wm_1p_2+m_1\wip_2)\big)\\  \\
 & = & \epsilon (m_1p_2-m_2p_1)\wm_3+\epsilon (m_2p_3-m_3p_2)\wm_1+\epsilon (m_3p_1-m_1p_3)\wm_2\\  \\
  & = & \epsilon (m_1p_2-m_2p_1)(\wm_3-m_3)+\epsilon (m_2p_3-m_3p_2)(\wm_1-m_1)+\epsilon (m_3p_1-m_1p_3)(\wm_2-m_2).
  \end{eqnarray*}
 Using the last three equations of motion \eqref{eq: dC}, we continue:
 \begin{eqnarray*}
& = &\epsilon^2(\omega_2-\omega_1)(m_1p_2-m_2p_1)(p_1\wip_2+p_2\wip_1)+\epsilon^2(\omega_3-\omega_2)(m_2p_3-m_3p_2)(p_2\wip_3+p_3\wip_2)\\
 & & +\epsilon^2(\omega_1-\omega_3)(m_3p_1-m_1p_3)(p_3\wip_1+p_1\wip_3)\\ \\
 & = & \;\;\big(\epsilon^2(\omega_2-\omega_1)p_2^2+\epsilon^2(\omega_3-\omega_1)p_3^2\big)m_1\wip_1
         +\big(\epsilon^2(\omega_2-\omega_1)\wip_2p_2+\epsilon^2(\omega_3-\omega_1)\wip_3p_3\big)m_1p_1\\ 
 & & +\big(\epsilon^2(\omega_3-\omega_2)p_3^2+\epsilon^2(\omega_1-\omega_2)p_1^2\big)m_2\wip_2
        +\big(\epsilon^2(\omega_3-\omega_2)\wip_3p_3+\epsilon^2(\omega_1-\omega_2)\wip_1p_1\big)m_2p_2 \\
 & & +\big(\epsilon^2(\omega_1-\omega_3)p_1^2+\epsilon^2(\omega_2-\omega_3)p_2^2\big)m_3\wip_3
       +\big(\epsilon^2(\omega_1-\omega_3)\wip_1p_1+\epsilon^2(\omega_2-\omega_3)\wip_2p_2\big)m_3p_3.
\end{eqnarray*}
This proves \eqref{eq: 2}. Again, this relation is the discrete time analog of $\sum_{i=1}^3m_i\dot{p}_i=0$, since $\sum_{i=1}^3 m_i(\wip_i-p_i)$ turns out to be of order $\epsilon^2$. 
\hfill $\blacksquare$
\smallskip

{\bf Proof of formula \eqref{eq: 3}.}
We start by using the first three equations of motion \eqref{eq: dC}:
\begin{eqnarray*}
\lefteqn{\sum_{i=1}^3 (\wm_i-m_i)\wip_i}\\ 
& = &\epsilon(\omega_3-\omega_2)(\wip_1\wip_3p_2+\wip_1\wip_2p_3)+\epsilon(\omega_1-\omega_3)(\wip_2\wip_1p_3+\wip_2\wip_3p_1)\\
 & & +\epsilon(\omega_2-\omega_1)(\wip_3\wip_2p_1+\wip_3\wip_1p_2)\\ \\
& = & \epsilon (\omega_1-\omega_2)\wip_1\wip_2p_3+\epsilon (\omega_2-\omega_3)\wip_2\wip_3p_1+\epsilon (\omega_3-\omega_1)\wip_3\wip_1p_2 \\ \\
& = & \epsilon (\omega_2-\omega_1)\wip_1\wip_2(\wip_3-p_3)+\epsilon (\omega_3-\omega_2)\wip_2\wip_3(\wip_1-p_1)
         +\epsilon (\omega_1-\omega_3)\wip_3\wip_1(\wip_2-p_2),
\end{eqnarray*}
and continue by using the last three equations of motion \eqref{eq: dC}:
\begin{eqnarray*}
& = & \epsilon^2 (\omega_2-\omega_1)\wip_1\wip_2\big((\wm_2p_1+m_2\wip_1)-(\wm_1p_2+m_1\wip_2)\big)\\
 & & +\epsilon^2 (\omega_3-\omega_2)\wip_2\wip_3\big((\wm_3p_2+m_3\wip_2)-(\wm_2p_3+m_2\wip_3)\big)\\
 & & +\epsilon^2 (\omega_1-\omega_3)\wip_3\wip_1\big((\wm_1p_3+m_1\wip_3)-(\wm_3p_1+m_3\wip_1)\big)\\ \\
 & = & \;\;\big(\epsilon^2(\omega_1-\omega_2)\wip_2^2+\epsilon^2(\omega_1-\omega_3)\wip_3^2\big)m_1\wip_1
          +\big(\epsilon^2(\omega_1-\omega_2)\wip_2p_2+\epsilon^2(\omega_1-\omega_3)\wip_3p_3\big)\wm_1\wip_1 \\
 & & +\big(\epsilon^2(\omega_2-\omega_3)\wip_3^2+\epsilon^2(\omega_2-\omega_1)\wip_1^2\big)m_2\wip_2
        +\big(\epsilon^2(\omega_2-\omega_3)\wip_3p_3+\epsilon^2(\omega_2-\omega_1)\wip_1p_1\big)\wm_2\wip_2\\
 & & +\big(\epsilon^2(\omega_3-\omega_1)\wip_1^2+\epsilon^2(\omega_3-\omega_2)\wip_2^2\big)m_3\wip_3
        +\big(\epsilon^2(\omega_3-\omega_1)\wip_1p_1+\epsilon^2(\omega_3-\omega_2)\wip_2p_2\big)\wm_3\wip_3.
\end{eqnarray*}
This proves \eqref{eq: 3}. Again, this relation is the discrete time analog of $\sum_{i=1}^3\dot{m}_ip_i=0$. 
\hfill $\blacksquare$
\smallskip

{\bf Proof of formula \eqref{eq: 4}.}
We compute:
\begin{eqnarray*}
\lefteqn{\sum_{i=1}^3 \wm_i(\wip_i-p_i)}\\ 
& = & \epsilon \wm_1\big((\wm_3p_2+m_3\wip_2)-(\wm_2p_3+m_2\wip_3)\big)   
         +\epsilon \wm_2\big((\wm_1p_3+m_1\wip_3)-(\wm_3p_1+m_3\wip_1)\big)\\
 & & +\epsilon \wm_3\big((\wm_2p_1+m_2\wip_1)-(\wm_1p_2+m_1\wip_2)\big)\\  \\
 & = & \epsilon (\wm_1\wip_2-\wm_2\wip_1)m_3+\epsilon (\wm_2\wip_3-\wm_3\wip_2)m_1+\epsilon (\wm_3\wip_1-\wm_1\wip_3)m_2\\  \\
  & = & \epsilon (\wm_2\wip_1-\wm_1\wip_2)(\wm_3-m_3)+\epsilon (\wm_3\wip_2-\wm_2\wip_3)(\wm_1-m_1)+\epsilon (\wm_1\wip_3-\wm_3\wip_1)(\wm_2-m_2)\\  \\
& = &\epsilon^2(\omega_2-\omega_1)(\wm_2\wip_1-\wm_1\wip_2)(p_1\wip_2+p_2\wip_1)+\epsilon^2(\omega_3-\omega_2)(\wm_3\wip_2-\wm_2\wip_3)(p_2\wip_3+p_3\wip_2)\\
 & & +\epsilon^2(\omega_1-\omega_3)(\wm_1\wip_3-\wm_3\wip_1)(p_3\wip_1+p_1\wip_3)\\ \\
 & = & \;\;\big(\epsilon^2(\omega_1-\omega_2)\wip_2^2+\epsilon^2(\omega_2-\omega_3)\wip_3^2\big)\wm_1p_1
         +\big(\epsilon^2(\omega_1-\omega_2)\wip_2p_2+\epsilon^2(\omega_1-\omega_3)\wip_3p_3\big)\wm_1\wip_1\\ 
 & & +\big(\epsilon^2(\omega_2-\omega_3)\wip_3^2+\epsilon^2(\omega_2-\omega_1)\wip_1^2\big)\wm_2p_2
        +\big(\epsilon^2(\omega_2-\omega_3)\wip_3p_3+\epsilon^2(\omega_2-\omega_1)\wip_1p_1\big)\wm_2\wip_2 \\
 & & +\big(\epsilon^2(\omega_3-\omega_1)\wip_1^2+\epsilon^2(\omega_3-\omega_2)\wip_2^2\big)\wm_3p_3
       +\big(\epsilon^2(\omega_3-\omega_1)\wip_1p_1+\epsilon^2(\omega_3-\omega_2)\wip_2p_2\big)\wm_3\wip_3.
\end{eqnarray*}
This proves \eqref{eq: 4} and provides us with another discrete time analog of $\sum_{i=1}^3m_i\dot{p}_i=0$. 
\hfill $\blacksquare$

\begin{corollary}\label{cor K}
The function
\begin{equation} \label{Clebsch1 dK}
K(m,p)=\sum_{i=1}^3 \frac{C_i}{C_0}\frac{m_ip_i}{c_0}
\end{equation}
is an integral of motion of the map $\Phi_f$.
\end{corollary}
\begin{proof}
Compare \eqref{eq: 1} with \eqref{eq: 4}, taking into account that $c_i/c_0$ are integrals of motion, that is, $c_i/c_0=\widetilde{c}_i/\widetilde{c}_0$. We arrive at
$$
\sum_{i=1}^3 C_i\frac{m_ip_i}{c_0}=\sum_{i=1}^3 C_i\frac{\wm_i\wip_i}{\widetilde{c}_0}.
$$
Since $C_i/C_0$ are integrals of motion, that is $C_i/C_0=\widetilde{C}_i/\widetilde{C}_0$, we see that
$$
\sum_{i=1}^3 \frac{C_i}{C_0}\frac{m_ip_i}{c_0}=\sum_{i=1}^3 \frac{\widetilde{C}_i}{\widetilde{C}_0}\frac{\wm_i\wip_i}{\widetilde{c}_0}.
$$
This proves the statement.
\end{proof}

\noindent
{\bf Proof of Theorem \ref{Th: Clebsch1 W2 basis}.}
We start with
\begin{eqnarray*}
\lefteqn{\sum_{i=1}^3 (\wwm_i-\wm_i)p_i+\sum_{i=1}^3(\wm_i-m_i)\wwp_i}\\
& = & \epsilon(\omega_3-\omega_2)(\wwp_2\wip_3+\wip_2\wwp_3)p_1+\epsilon(\omega_1-\omega_3)(\wwp_3\wip_1+\wip_3\wwp_1)p_2
 +\epsilon(\omega_2-\omega_1)(\wwp_1\wip_2+\wip_1\wwp_2)p_3\\
 & & +\epsilon(\omega_3-\omega_2)(\wip_2p_3+p_2\wip_3)\wwp_1+\epsilon(\omega_1-\omega_3)(\wip_3p_1+p_3\wip_1)\wwp_2
 +\epsilon(\omega_2-\omega_1)(\wip_1p_2+p_1\wip_2)\wwp_3\\ \\
 & = & \epsilon(\omega_3-\omega_1)\wwp_1\wip_2p_3+\epsilon(\omega_1-\omega_2)\wwp_1\wip_3p_2+
\epsilon(\omega_1-\omega_2)\wwp_2\wip_3p_1+\epsilon(\omega_2-\omega_3)\wwp_2\wip_1p_3\\
& & +
\epsilon(\omega_2-\omega_3)\wwp_3\wip_1p_2+\epsilon(\omega_3-\omega_1)\wwp_3\wip_2p_1\\ \\
& = & \epsilon(\omega_3-\omega_2) \wwp_1p_2p_3 +\epsilon (\omega_1-\omega_3) \wwp_2p_3p_1+\epsilon (\omega_2-\omega_1) \wwp_3p_1p_2\\
 &  & +\epsilon(\omega_3-\omega_1)\wwp_1(\wip_2-p_2)p_3+\epsilon(\omega_1-\omega_2)\wwp_1(\wip_3-p_3)p_2+
\epsilon(\omega_1-\omega_2)\wwp_2(\wip_3-p_3)p_1\\
& & +\epsilon(\omega_2-\omega_3)\wwp_2(\wip_1-p_1)p_3+\epsilon(\omega_2-\omega_3)\wwp_3(\wip_1-p_1)p_2+
\epsilon(\omega_3-\omega_1)\wwp_3(\wip_2-p_2)p_1
\end{eqnarray*}
This is equal to 
\begin{eqnarray*}
& = & \epsilon(\omega_3-\omega_2) (\wwp_1-\wip_1)p_2p_3 +\epsilon (\omega_1-\omega_3) (\wwp_2-\wip_2)p_3p_1+\epsilon (\omega_2-\omega_1) (\wwp_3-\wip_3)p_1p_2\\
 & & + \epsilon(\omega_3-\omega_2) \wip_1p_2p_3 +\epsilon (\omega_1-\omega_3) \wip_2p_3p_1+\epsilon (\omega_2-\omega_1) \wip_3p_1p_2\\
 &  & +\epsilon(\omega_3-\omega_1)\wwp_1(\wip_2-p_2)p_3+\epsilon(\omega_1-\omega_2)\wwp_1(\wip_3-p_3)p_2+
\epsilon(\omega_1-\omega_2)\wwp_2(\wip_3-p_3)p_1\\
& & +\epsilon(\omega_2-\omega_3)\wwp_2(\wip_1-p_1)p_3+\epsilon(\omega_2-\omega_3)\wwp_3(\wip_1-p_1)p_2+
\epsilon(\omega_3-\omega_1)\wwp_3(\wip_2-p_2)p_1.
\end{eqnarray*}
Here the second line is equal to
$$
-\sum_{i=1}^3 (\wm_i-m_i)p_i,
$$
so that, upon use of equations of motion, we find:
\begin{eqnarray*}
\lefteqn{\sum_{i=1}^3(\wwm_ip_i-m_i\wwp_i)-\sum_{i=1}^3 m_ip_i+\sum_{i=1}^3 \wm_i\wwp_i}\\
& = & \epsilon^2(\omega_3-\omega_2) \big((\wwm_3\wip_2+\wm_3\wwp_2)-(\wwm_2\wip_3+\wm_2\wwp_3)\big)p_2p_3 \\
 & & +\epsilon^2 (\omega_1-\omega_3) \big((\wwm_1\wip_3+\wm_1\wwp_3)-(\wwm_3\wip_1+\wm_3\wwp_1)\big)p_3p_1\\
 & & +\epsilon^2 (\omega_2-\omega_1) \big((\wwm_2\wip_1+\wm_2\wwp_1)-(\wwm_1\wip_2+\wm_1\wwp_2)\big)p_1p_2\\
 &  & +\epsilon^2(\omega_3-\omega_1)\wwp_1\big((\wm_1p_3+m_1\wip_3)-(\wm_3p_1+m_3\wip_1)\big)p_3\\
 & & +\epsilon^2(\omega_1-\omega_2)\wwp_1\big((\wm_2p_1+m_2\wip_1)-(\wm_1p_2+m_1\wip_2)\big)p_2\\
& & +\epsilon^2(\omega_1-\omega_2)\wwp_2\big((\wm_2p_1+m_2\wip_1)-(\wm_1p_2+m_1\wip_2)\big)p_1\\
& & +\epsilon^2(\omega_2-\omega_3)\wwp_2\big((\wm_3p_2+m_3\wip_2)-(\wm_2p_3+m_2\wip_3)\big)p_3\\
& & +\epsilon^2(\omega_2-\omega_3)\wwp_3\big((\wm_3p_2+m_3\wip_2)-(\wm_2p_3+m_2\wip_3)\big)p_2\\
& & +\epsilon^2(\omega_3-\omega_1)\wwp_3\big((\wm_1p_3+m_1\wip_3)-(\wm_3p_1+m_3\wip_1)\big)p_1.
\end{eqnarray*}
(So, the strategy of transformations is: leave one variable with double tilde, one with single tilde and two without tilde.) Collecting terms, we have:
\begin{eqnarray*}
\lefteqn{\sum_{i=1}^3(\wwm_ip_i-m_i\wwp_i)-\sum_{i=1}^3 m_ip_i+\sum_{i=1}^3 \wm_i\wwp_i}\\
& = & \wwm_1p_1\Big(\epsilon^2(\omega_1-\omega_3)\wip_3p_3+\epsilon^2(\omega_1-\omega_2)\wip_2p_2\Big) \\
& & +\wwm_2p_2\Big(\epsilon^2(\omega_2-\omega_3)\wip_3p_3+\epsilon^2(\omega_2-\omega_1)\wip_1p_1\Big) \\
& & +\wwm_3p_3\Big(\epsilon^2(\omega_3-\omega_2)\wip_2p_2+\epsilon^2(\omega_3-\omega_1)\wip_1p_1\Big) \\
&  & +m_1\wwp_1\Big(\epsilon^2(\omega_3-\omega_1)\wip_3p_3+\epsilon^2(\omega_2-\omega_1)\wip_2p_2\Big) \\
& & +m_2\wwp_2\Big(\epsilon^2(\omega_1-\omega_2)\wip_1p_1+\epsilon^2(\omega_3-\omega_2)\wip_3p_3\Big) \\
& & +m_3\wwp_3\Big(\epsilon^2(\omega_2-\omega_3)\wip_2p_2+\epsilon^2(\omega_1-\omega_3)\wip_1p_1\Big) \\
&  & +m_1p_1\Big(\epsilon^2(\omega_2-\omega_1)\wwp_2\wip_2+\epsilon^2(\omega_3-\omega_1)\wwp_3\wip_3\Big) \\
& & +m_2p_2\Big(\epsilon^2(\omega_1-\omega_2)\wwp_1\wip_1+\epsilon^2(\omega_3-\omega_2)\wwp_3\wip_3\Big) \\
& & +m_3p_3\Big(\epsilon^2(\omega_1-\omega_3)\wwp_1\wip_1+\epsilon^2(\omega_2-\omega_3)\wwp_2\wip_2\Big) \\
&  & +\wm_1\wwp_1\Big(\epsilon^2(\omega_2-\omega_1)p_2^2+\epsilon^2(\omega_3-\omega_1)p_3^2\Big) \\
& & +\wm_2\wwp_2\Big(\epsilon^2(\omega_1-\omega_2)p_1^2+\epsilon^2(\omega_3-\omega_2)p_3^2\Big) \\
& & +\wm_3\wwp_3\Big(\epsilon^2(\omega_2-\omega_3)p_2^2+\epsilon^2(\omega_1-\omega_3)p_1^2\Big) .
\end{eqnarray*}
This can be put as
$$
\sum_{i=1}^3 C_i(\wwm_ip_i-m_i\wwp_i)=\sum_{i=1}^3 \widetilde{C}_im_ip_i-\sum_{i=1}^3 c_i\wm_i\wwp_i.
$$
From \eqref{Clebsch1 dK}, \eqref{eq: 3}, and from relations $\widetilde{C}_i/\widetilde{C}_0=C_i/C_0$ and $c_i/c_0=\widetilde{c}_i/\widetilde{c}_0$, we derive:
$$
\sum_{i=1}^3 \widetilde{C}_im_ip_i=\widetilde{C}_0\cdot c_0 K(m,p), \quad \sum_{i=1}^3 c_i\wm_i\wwp_i=c_0\cdot \widetilde{C}_0 K(\wm,\wip).
$$
By Corollary \ref{cor K}, $K(m,p)=K(\wm,\wip)$. This finishes the proof.
\hfill $\blacksquare$
\smallskip

As for Theorem \ref{Clebsch higher Wronskians}, at present we only have a proof based on symbolic computations by Maple, even for the first Clebsch flow.

\section{The Kirchhoff case}
\label{sect Kirchhoff}

The {\em Kirchhoff case} of the motion of the rigid body in an ideal fluid corresponds to the following values of the parameters in \eqref{gClebsch Ham}, \eqref{gClebsch}:
\begin{equation}\label{eq: Kirchhoff cond}
a_1=a_2, \quad b_1=b_2.
\end{equation}
Equations of motion read:
\begin{equation} \label{eq:Kirchhoff}
\left\{\begin{array}{l}
\dot{m_1} = (a_3-a_1)m_2 m_3 + (b_3-b_1)p_2p_3, \vspace{.1truecm}  \\
\dot{m_2} = (a_1-a_3)m_1 m_3 + (b_1-b_3)p_1p_3, \vspace{.1truecm}  \\
\dot{m_3} = 0, \vspace{.1truecm}  \\
\dot{p_1} =  a_3p_2 m_3-a_1p_3 m_2, \vspace{.1truecm}  \\
\dot{p_2} =  a_1p_3 m_1-a_3p_1m_3,  \vspace{.1truecm}  \\
\dot{p_3} = a_1(p_1m_2 - p_2m_1).
\end{array}\right.
\end{equation}
Thus, $m_3$ is an obvious fourth integral, due to the rotational symmetry of the system. It is easy to see that for (\ref{eq: Kirchhoff cond}) the Clebsch condition \eqref{eq: Clebsch cond} is satisfied, as well. Thus, formally the Kirchhoff case is the particular case of the Clebsch case. One can choose the parameters $\omega_i$ in \eqref{eq: Clebsch omega} as
\begin{equation}\label{eq: Kirchhoff omega}
\omega_1=\omega_2=\frac{b_1}{a_1}, \quad \omega_3=\frac{b_3}{a_1}.
\end{equation}
Correspondingly, integral $H_1$ becomes proportional to
$$
a_1H_1=a_1(m_1^2+m_2^2+m_3^2)+b_1(p_1^2+p_2^2)+b_3p_3^2.
$$
Taking into account that the Hamilton function $H$ is given by 
$$
2H=a_1(m_1^2+m_2^2)+a_3m_3^2+b_1(p_1^2+p_2^2)+b_3p_3^2,
$$
we see that the fourth integral $H_1$ can be replaced just by $m_3^2$.  

Wronskian relation satisfied on solutions of \eqref{eq:Kirchhoff}:
\begin{equation}\label{eq: KC W}
(\dot{m}_1p_1-m_1\dot{p}_1)+(\dot{m}_2p_2-m_2\dot{p}_2)+\left(\frac{2a_3}{a_1}-1\right)(\dot{m}_3p_3-m_3\dot{p}_3)=0.
\end{equation}

Applying the Kahan-Hirota-Kimura scheme to the Kirchhoff system \eqref{eq:Kirchhoff}, we arrive at the following discretization:
\beq\label{eq:dKirchhoff}
 \left\{ \begin{array}{l}
\widetilde{m}_1-m_1 =
\epsilon(a_3-a_1)(\widetilde{m}_2m_3+m_2\widetilde{m}_3)+
\epsilon(b_3-b_1)(\widetilde{p}_2p_3+p_2\widetilde{p}_3),
\vspace{.1truecm} \\
\widetilde{m}_2-m_2 =
\epsilon(a_1-a_3)(\widetilde{m}_3m_1+m_3\widetilde{m}_1)+
\epsilon(b_1-b_3)(\widetilde{p}_3p_1+p_3\widetilde{p}_1),
\vspace{.1truecm}\\
\widetilde{m}_3-m_3 =0,
\vspace{.1truecm} \\
\widetilde{p}_1-p_1 = \epsilon
a_3(\widetilde{m}_3p_2+m_3\widetilde{p}_2)-\epsilon a_1(\widetilde{m}_2 p_3+m_2\widetilde{p}_3),
\vspace{.1truecm}\\
\widetilde{p}_2-p_2= \epsilon
a_1(\widetilde{m}_1p_3+m_1\widetilde{p}_3)-\epsilon a_3(\widetilde{m}_3p_1+m_3\widetilde{p}_1),
\vspace{.1truecm} \\
\widetilde{p}_3-p_3 = \epsilon
a_1(\widetilde{m}_2p_1+m_2\widetilde{p}_1-\widetilde{m}_1 p_2-m_1\widetilde{p}_2).
\end{array}\right.
\end{equation}
As usual, linear system (\ref{eq:dKirchhoff}) defines a birational map $\Phi_f:\bbR^6 \rightarrow \bbR^6$, $(m,p)\mapsto (\wm,\wip)$.

\begin{theorem}\label{Th: dKirchhoff basis 1}  {\bf (Quadratic-fractional integral, \cite{PPS2})}
The function
\begin{equation} \label{dKirchhoff simple int}
I_0(m,p;\epsilon)=\frac{c_3(m,p;\epsilon)}{c_1(m,p;\epsilon)},
\end{equation}
where
\begin{eqnarray}
c_1 & = & 1+\epsilon^2a_3(a_1-a_3)m_3^2+\epsilon^2a_1(b_1-b_3)p_3^2,                                                         \label{eq: dKirchhoff c1}\\
c_3 & = & \frac{2a_3}{a_1}-1+\epsilon^2a_1(a_3-a_1)(m_1^2+m_2^2)+\epsilon^2a_3(b_3-b_1)(p_1^2+p_2^2), \label{eq: dKirchhoff c3}
\end{eqnarray}
is an integral of motion of the map $\Phi_f$. 
\end{theorem}

\begin{theorem} \label{Th: dKirchhoff W1 basis} {\bf (Discrete Wronskians HK basis, \cite{PPS2})}
Functions $W^{(1)}_i(m,p)=\wm_ip_i-m_i\wip_i$, $i=1,2,3$, form a HK basis for the map $\Phi_f$ with a one-dimensional null space spanned by
$[c_1:c_1:c_3]=[1:1:I_0]$, where $I_0$ is the integral of $\Phi_f$ given by \eqref{dKirchhoff simple int}.
\end{theorem}

Novel results, illustrating Observations \ref{conjecture 1} and \ref{conjecture 2}, are as follows.

\begin{theorem}\label{Th: dKirchhoff basis 2} {\bf (Bilinear-fractional integral)}
The function
\begin{equation} \label{dKirchhoff 2nd int}
J_0(m,p;\epsilon)=\frac{C_3(m,p;\epsilon)}{C_1(m,p;\epsilon)},
\end{equation}
where
\begin{eqnarray}
C_1 & = & 1-\epsilon^2a_3(a_1-a_3)m_3^2-\epsilon^2a_1(b_1-b_3)p_3\wip_3,     \label{eq: dKirchhoff C1}\\
C_3 & = & \frac{2a_3}{a_1}-1-\epsilon^2a_1(a_3-a_1)(m_1\wm_1+m_2\wm_2)-\epsilon^2a_3(b_3-b_1)(p_1\wip_1+p_2\wip_2),\qquad
\label{eq: dKirchhoff R}
\end{eqnarray}
is an integral of motion of the map $\Phi_f$. 
\end{theorem}

\begin{corollary}\label{Th: Kirchhoff conserved density} {\bf (Density of an invariant measure)}
The map $\Phi_f(x;\epsilon)$ has an invariant measure
$$
\frac{dm_1\wedge dm_2\wedge dm_3 \wedge dp_1\wedge dp_2\wedge dp_3}{\phi(m,p;\epsilon)},
$$
where for $\phi(m,p;\epsilon)$ one can take the numerator of either of the functions $C_1$, $C_3$.
\end{corollary}

\begin{theorem}  \label{Th: dKirchhoff W2 basis} {\bf (Second order Wronskians HK basis)}
The functions $$W^{(2)}_i(m,p)=\widetilde{\wm}_ip_i-m_i\widetilde{\wip}_i, \quad i=1,2,3,$$ form a HK basis for the map $\Phi_f$, with a one-dimensional null space spanned by
$[C_1:C_1:C_3]=[1:1:J_0]$, with the function $J_0$ given in \eqref{dKirchhoff 2nd int}. 
\end{theorem}

\begin{theorem} \label{Kirchhoff higher Wronskians} {\bf (Third order Wronskians HK basis)}
Functions $W^{(3)}_i(m,p)=\widetilde{\widetilde {\widetilde{m}}}_ip_i-m_i\widetilde{\widetilde {\widetilde{p}}}_i$, $i=1,2,3$, form a HK basis for the map $\Phi_f$, with a one-dimensional null space.
On orbits of the map $\Phi_f$ there holds
$$
(\widetilde{\widetilde {\widetilde{m}}}_1p_1-m_1\widetilde{\widetilde {\widetilde{p}}}_1) +
(\widetilde{\widetilde {\widetilde{m}}}_2p_2-m_2\widetilde{\widetilde {\widetilde{p}}}_2) + J_1(\widetilde{\widetilde {\widetilde{m}}}_3p_3-m_3\widetilde{\widetilde {\widetilde{p}}}_3) =0,
$$
where the function $J_1(m,p;\epsilon)$ is an integral of motion. The four integrals of motion 
$\{I_0,J_0,J_1,m_3\}$ are functionally independent .
\end{theorem}

\section{Lagrange top}
\label{sect Lagrange}

The Hamilton function of the Lagrange top is $H=\frac{1}{2}H_1$, where
\beq \label{eq: lagrangehamilton}
H_1=m_1^2+m_2^2+\alpha m_3^2+2\gamma p_3.
\eeq
Unlike the Clebsch and the Kirchhoff cases, this function is not homogeneous.
Equations of motion of Lagrange top read
\begin{equation}  \label{eq: lagrange}
\left\{\begin{array}{l}
\dot{m}_1  =  (\alpha-1)m_2m_3 + \gamma p_2, \vspace{.1truecm} \\
\dot{m}_2  = (1-\alpha)m_1m_3 - \gamma p_1, \vspace{.1truecm} \\
\dot{m}_3  = 0, \vspace{.1truecm} \\
\dot{p}_1 = \alpha p_2m_3-p_3m_2, \vspace{.1truecm} \\
\dot{p}_2 = p_3m_1-\alpha p_1m_3, \vspace{.1truecm} \\
\dot{p}_3 = p_1m_2-p_2m_1.
\end{array}\right.
\end{equation}
So, like in the Kirchhoff case, $m_3$ is an obvious fourth integral, due to the rotational symmetry of the system. 

Wronskian relation satisfied on solutions of \eqref{eq: lagrange}:
\begin{equation}\label{eq: LT W}
(\dot{m}_1p_1-m_1\dot{p}_1)+(\dot{m}_2p_2-m_2\dot{p}_2)+(2\alpha-1)(\dot{m}_3p_3-m_3\dot{p}_3)=0.
\end{equation}

Applying the Kahan-Hirota-Kimura scheme to the vector field $f$ from \eqref{eq: lagrange}, we arrive at a birational map $\Phi_f:\bbR^6 \rightarrow \bbR^6$, $(m,p)\mapsto (\wm,\wip)$, defined by the following linear system:
\beq\label{eq:dLagrange}
 \left\{ \begin{array}{l}
\widetilde{m}_1-m_1 = \epsilon(\alpha-1)(\widetilde{m}_2m_3+m_2\widetilde{m}_3)+\epsilon\gamma(p_2+\widetilde{p}_2),  \vspace{.1truecm} \\
\widetilde{m}_2-m_2 = \epsilon(1-\alpha)(\widetilde{m}_3m_1+m_3\widetilde{m}_1)- \epsilon\gamma(p_1+\widetilde{p}_1),  \vspace{.1truecm}\\
\widetilde{m}_3-m_3 =0, \vspace{.1truecm} \\
\widetilde{p}_1-p_1 = \epsilon \alpha (p_2\widetilde{m}_3+\widetilde{p}_2m_3) - \epsilon (p_3\widetilde{m}_2+\widetilde{p}_3m_2), \vspace{.1truecm}\\
\widetilde{p}_2-p_2= \epsilon (p_3\widetilde{m}_1+\widetilde{p}_3m_1) - \epsilon \alpha (p_1\widetilde{m}_3+\widetilde{p}_1m_3),  \vspace{.1truecm}\\
\widetilde{p}_3-p_3 = \epsilon (p_1\widetilde{m}_2+\widetilde{p}_1m_2-p_2\widetilde{m}_1-\widetilde{p}_2m_1).
\end{array}\right.
\end{equation}

\begin{theorem}\label{Th: dLagrange basis 1}  {\bf (Quadratic-fractional integral, \cite{PPS2})}
\begin{equation} \label{dLagrange simple int}
I_0(m,p;\epsilon)=\frac{r(m,p;\epsilon)}{s(m,p;\epsilon)},
\end{equation}
where
\begin{eqnarray}
r & = & (2\alpha-1)+\epsilon^2(\alpha-1)(m_1^2+m_2^2)+\frac{\epsilon^2\gamma}{m_3}(m_1p_1+m_2p_2),
\label{eq: dLagrange r}\\
s & = & 1+\epsilon^2\alpha(1-\alpha)m_3^2-\epsilon^2\gamma p_3,\label{eq: dLagrange s}
\end{eqnarray}
is an integral of motion of the map $\Phi_f$. 
\end{theorem}

\begin{theorem} \label{Th: dLagrange W1 basis} {\bf (Discrete Wronskians HK basis, \cite{PPS2})}
Functions $W^{(1)}_i(m,p)=\wm_ip_i-m_i\wip_i$, $i=1,2,3$, form a HK basis for the map $\Phi_f$ with a one-dimensional null space spanned by
$[1:1:I_0]$, where $I_0$ is the integral of $\Phi_f$ given by \eqref{dLagrange simple int}.
\end{theorem}

Novel results, supporting Observations \ref{conjecture 1} and \ref{conjecture 2}, are as follows.

\begin{theorem}\label{Th: dLagrange basis 2} {\bf (Bilinear-fractional integral)}
The function
\begin{equation} \label{dLagrange 2nd int}
J_0(m,p;\epsilon)=\frac{R(m,p;\epsilon)}{S(m,p;\epsilon)},
\end{equation}
where
\begin{eqnarray}
R & = & (2\alpha-1)-\epsilon^2(\alpha-1)(m_1\wm_1+m_2\wm_2)-\frac{\epsilon^2\gamma}{2m_3}(\wm_1p_1+m_1\wip_1+\wm_2p_2+m_2\wip_2), \nonumber\\
\label{eq: dLagrange R}\\
S & = & 1-\epsilon^2\alpha(1-\alpha)m_3^2+\tfrac{1}{2}\epsilon^2\gamma (p_3+\wip_3),  
\label{eq: dLagrange S}
\end{eqnarray}
is an integral of motion of the map $\Phi_f$. 
\end{theorem}

\begin{corollary}\label{Th: dLagrange conserved density} {\bf (Density of an invariant measure)}
The map $\Phi_f(x;\epsilon)$ has an invariant measure
$$
\frac{dm_1\wedge dm_2\wedge dm_3 \wedge dp_1\wedge dp_2\wedge dp_3}{\phi(m,p;\epsilon)},
$$
where $\phi(m,p;\epsilon)$ can be taken as the numerator of either of the functions $R$, $S$.
\end{corollary}

\begin{theorem}  \label{Th: dLagrange W2 basis} {\bf (Second order Wronskians HK basis)}
The functions $$W^{(2)}_i(m,p)=\widetilde{\wm}_ip_i-m_i\widetilde{\wip}_i, \quad i=1,2,3,$$ form a HK basis for the map $\Phi_f$, with a one-dimensional null space spanned by
$[1:1:J_0]$, with the function $J_0$ given in \eqref{dLagrange 2nd int}. 
\end{theorem}

There holds a theorem which reads literally as Theorem \ref{Kirchhoff higher Wronskians} on the third order Wronskians HK basis.

\section{Concluding remarks}

We would like to remark that the existence of quadratic-fractional integrals of the Kahan discretizations is a rather common phenomenon which even is not related to integrability. In this connection, we refer to the recent paper \cite{CMOQ4}, where the following result is established.
\begin{theorem}\label{th CQ}
Let two components a quadratic vector field be of the form 
\begin{equation} \label{flow gen gen}
\left\{
\begin{array}{l}
\dot x_1 =  \ell(x)\dfrac{\partial H}{\partial x_2}=\ell(x)(bx_1+cx_2), \vspace{.2truecm} \\
\dot x_2  = -\ell(x)\dfrac{\partial H}{\partial x_1}=-\ell(x)(ax_1+bx_2), 
\end{array}
\right. 
\end{equation}
where $\ell(x)$ is an affine function on $\mathbb R^n$, and 
\begin{equation}\label{H hom}
H(x_1,x_2)= \frac{1}{2}(ax_1^2+2bx_1x_2+cx_2^2)
\end{equation}
is a quadratic form of $x_1$, $x_2$. The Kahan discretization of this vector field, with the first two equations of motion
\begin{equation}\label{map gen gen}
\left\{
\begin{array}{l}
\tilde x_1- x_1 = \epsilon \ell(x)(b\wx_1+c\wx_2) + \epsilon \ell(\wx) (bx_1+cx_2) , \vspace{.2truecm} \\
\tilde x_2 - x_2 = -\epsilon \ell(x) (a\wx_1+b\wx_2) -\epsilon \ell(\wx) (ax_1+bx_2),
\end{array}
\right. 
\end{equation}
admits an quadratic-fractional integral of motion 
\begin{equation}
F(x,\ep)= \frac{ax_1^2+2bx_1x_2+cx_2^2}{1 +\ep^2 (ac-b^2) \ell^2(x)}.
\end{equation}
\end{theorem}
Actually, their result holds true for any (not necessarily homogeneous) quadratic polynomial $H(x_1,x_2)$, but this generalization easily follows by shift of variables. This result does not depend on equations of motion for $x_k$ with $3\le k\le n$, and therefore it is unrelated to integrability. It turns out that the procedure described in Observation \ref{conjecture 1} (polarization of the quadratic polynomials in the numerator and in the denominator, accompanied by the change $\epsilon^2\to -\epsilon^2$) works for the whole class of vector fields described in Theorem \ref{th CQ}, but, amazingly, it does not lead to new integrals. We have the following result.
\begin{theorem} 
On orbits of the map \eqref{map gen gen}, we have: $\widehat F(x,\ep)= F(x,\ep)$, where
\begin{equation}
\widehat F(x,\ep)= \frac{ax_1\wx_1+b(x_1\wx_2+\wx_1x_2)+cx_2\wx_2}{1 - \epsilon^2 (ac-b^2) \ell(x)\ell(\wx)} .
\end{equation}
\end{theorem}
\begin{proof} Relation
$$
\frac{ax_1\wx_1+b(x_1\wx_2+x_2\wx_1)+cx_2\wx_2}{1-\epsilon^2 (ac-b^2) \ell(\tilde x)\ell(x)}=\frac{ax_1^2+2bx_1x_2+cx_2^2}{1+\epsilon^2 (ac-b^2) \ell^2(x)}
$$
is equivalent to
\begin{eqnarray*}
&&ax_1(\wx_1-x_1)+bx_1(\wx_2-x_2)+bx_2(\wx_1-x_1)+cx_2(\wx_2-x_2)\\
&&=-\ep^2(ac-b^2)\ell(x)\Big(ax_1\big(\wx_1\ell(x)+x_1\ell(\wx)\big)+bx_1\big(\wx_2\ell(x)+x_2\ell(\wx)\big)\\
&&\quad+bx_2\big(\wx_1\ell(x)+x_1\ell(\wx)\big)+cx_2\big(\wx_2\ell(x)+x_2\ell(\wx)\big)\Big).
\end{eqnarray*}
We transform the left-hand side, using equations of motion \eqref{map gen gen}:
\begin{eqnarray*}
&&(ax_1+bx_2)(\wx_1-x_1)+(bx_1+cx_2)(\wx_2-x_2)\\
&&=\epsilon(ax_1+bx_2)\Big((b\wx_1+c\wx_2)\ell(x) + (bx_1+cx_2) \ell(\wx)\Big)\\
&& \quad -\epsilon(bx_1+cx_2)\Big((a\wx_1+b\wx_2)\ell(x) + (ax_1+bx_2) \ell(\wx)\Big)\\
&& =\epsilon(ac-b^2)(x_1\wx_2-x_2\wx_1)\ell(x).
\end{eqnarray*}
A similar transformation of the right-hand side leads to:
\begin{eqnarray*}
&& -\ep^2(ac-b^2)\ell(x)\Big(x_1\big((a\wx_1+b\wx_2)\ell(x)+(ax_1+bx_2)\ell(\wx)\big)\\
&& \quad+x_2\big((b\wx_1+c\wx_2)\ell(x)+(bx_1+cx_2)\ell(\wx)\big)\Big)\\
&& = \epsilon (ac-b^2)\ell(x)\Big(x_1(\wx_2-x_2)-x_2(\wx_1-x_1)\Big)\\
&& = \epsilon (ac-b^2)\ell(x)(x_1\wx_2-x_2\wx_1).
\end{eqnarray*}
This proves the theorem.
\end{proof}

The latter result makes the applicability of Observation \ref{conjecture 1} all the more intriguing. In the majority of cases when the Kahan discretization possesses a quadratic-fractional integral of motion $F(x,\epsilon)$, the polarization of the latter (i.e., the function $\widehat F(x,\ep)$ obtained by the polarization of  the numerator and of the denominator of $F(x,\ep)$, accompanied by the change $\ep^2\to -\ep^2$) turns out to be an integral, as well. Further examples of this observation are delivered by integrable systems with the Lax representation $\dot{L}=[L^2,A]$ in the following situations: 
\begin{itemize}
\item$L$ a $3\times 3$ or a $4\times 4$ skew-symmetric matrix and $A$ a constant diagonal matrix (Euler top on the algebras $so(3)$ and $so(4)$);
\item $L$ a symmetric $3\times 3$ matrix and $A$ is a constant skew-symmetric $3\times 3$ matrix  \cite{BI, BBIMR};
\item $L$ a $3\times 3$ matrix with vanishing diagonal and $A$ a constant diagonal matrix (3-wave system, see \cite{PPS2}).
\end{itemize}
The only counterexample we are aware of, is given by the system $\dot{L}=[L^2,A]$ with a general $3\times 3$ matrix $L$ and a constant diagonal matrix $A$ \cite{AI}. For this system, polarization applied to quadratic-fractional integrals of the Kahan-Hirota-Kimura discretization (there are two independent such integrals) does not lead to integrals of motion.

Clarifying all the mysterious observations related to the Kahan-Hirota-Kimura discretization remains an intriguing and a rewarding task.

\section*{Acknowledgements}
This research is supported by the DFG Collaborative Research Center TRR 109 ``Discretization in Geometry and Dynamics''.

\end{document}